\PassOptionsToPackage{prologue,dvipsnames,svgnames}{xcolor}

\documentclass[letterpaper]{IEEEtran}
\usepackage{cite}
\usepackage{amsmath,amssymb,amsfonts}
\usepackage{algorithm}
\usepackage{algpseudocode}
\usepackage{graphicx}
\usepackage{textcomp}
\usepackage{amsthm}
\usepackage{kotex}
\usepackage[dvipsnames,svgnames]{xcolor}
\usepackage{array}
\usepackage{comment}
\usepackage{enumitem}
\usepackage{caption}
\usepackage{subcaption}
\usepackage{lipsum}
\usepackage{url}
\usepackage{adjustbox}
\usepackage{hyperref}

\usepackage{tikz}
\usepackage{pgfplots}
\usepackage{grffile}
\usetikzlibrary{shapes.multipart}
\usetikzlibrary{shapes,arrows}
\usetikzlibrary{arrows,calc,positioning}
\tikzset{
	block/.style = {draw, rectangle,
		minimum height=1cm,
		minimum width=2cm},
	input/.style = {coordinate,node distance=1cm},
	output/.style = {coordinate,node distance=4cm},
	arrow/.style={draw, -latex,node distance=2cm},
	pinstyle/.style = {pin edge={latex-, black,node distance=2cm}},
	sum/.style = {draw, circle, node distance=1cm},
}

\pgfplotsset{compat=newest}
\usetikzlibrary{plotmarks}
\usetikzlibrary{arrows.meta}
\usepgfplotslibrary{patchplots}

\newtheorem{theorem}{Theorem}

\newtheorem{proposition}{Proposition}

\newtheorem{remark}{Remark}

\newtheorem{lemma}{Lemma}
\newtheorem{problem}{Problem}
\newtheorem{definition}{Definition}

\newcommand{\Z}{\ensuremath{{\mathbb Z}}}
\newcommand{\N}{\ensuremath{{\mathbb N}}}
\newcommand{\R}{\ensuremath{{\mathbb R}}}

\newcommand{\calT}{\ensuremath{\mathcal{T}}}
\newcommand{\calF}{\ensuremath{\mathcal{F}}}
\newcommand{\calI}{\ensuremath{\mathcal{I}}}

\DeclareMathOperator{\col}{col}

\newcommand{\ini}{\ensuremath{\mathrm{ini}}}

\def\BibTeX{{\rm B\kern-.05em{\sc i\kern-.025em b}\kern-.08em
		T\kern-.1667em\lower.7ex\hbox{E}\kern-.125emX}}

\makeatletter
\def\endthebibliography{%
	\def\@noitemerr{\@latex@warning{Empty `thebibliography' environment}}%
	\endlist
}
\makeatother

\begin{document}
	
	\title{Stabilization by Controllers Having Integer Coefficients}
	\author{Joowon Lee, Donggil Lee, and Junsoo Kim
		\thanks{This work was supported by the National Research Foundation of Korea(NRF) grant funded by the Korea government(MSIT) (No. RS-2022-00165417 and No. RS-2024-00353032).}
		\thanks{J.~Lee is with ASRI, the Department of Electrical and Computer Engineering, Seoul National University, Seoul 08826, South Korea (e-mail: jwlee@cdsl.kr).}
		\thanks{D.~Lee is with the Department of Electrical Engineering, Incheon National University, Incheon 22012, South Korea (e-mail: dglee@inu.ac.kr).}
		\thanks{J.~Kim is with the Department of Electrical and Information Engineering, Seoul National University of Science and Technology, Seoul 01811, South Korea (e-mail: junsookim@seoultech.ac.kr).}}
	
	\maketitle
	
	\begin{abstract}

		The system property of ``having integer coefficients,'' that is, a transfer function has an integer monic polynomial as its denominator, is significant in the field of encrypted control as it is required for a dynamic controller to be realized over encrypted data. This paper shows that there always exists a controller with integer coefficients stabilizing a given discrete-time linear time-invariant plant. A constructive algorithm to obtain such a controller is provided, along with numerical examples. Furthermore, the proposed method is applied to converting a pre-designed controller to have integer coefficients, while the original performance is preserved in the sense that the transfer function of the closed-loop system remains unchanged.

	\end{abstract}
	
	
	\section{Introduction}\label{sec:intro}
	
	This paper addresses the classical problem of designing a controller that stabilizes a given plant, but under an additional constraint that the controller consists of integer coefficients.
	What we mean by \emph{consisting of integer coefficients} is that the denominator of a transfer function is an integer monic polynomial, that is, every coefficient of the denominator is an integer when its leading coefficient is one.
	To be specific, consider a single-input single-output (SISO) linear discrete-time plant described by a proper transfer function
	\begin{equation}\label{eq:plant}
		P(z)=\frac{N_p(z)}{D_p(z)},
	\end{equation}
	where the denominator $D_p(z)$ and the numerator $N_p(z)$ are coprime.
	Then,
	the problem is to find a controller
	\begin{equation}\label{eq:ctr}
		C(z)=\frac{N_c(z)}{D_c(z)}
	\end{equation}
	such that the polynomial
	\begin{equation}\label{eq:cl}
		D_p(z)D_c(z)-N_p(z)N_c(z)
	\end{equation}
	is Schur stable while $D_c(z)$ is an integer monic polynomial.
	
	The need for such controllers with integer coefficients has been prominent in the field of encrypted control \cite{CSM,ARC,MoritzARC,Kogiso,VC}, where modern cryptography is applied to networked control systems in a way that control operations are implemented directly over encrypted data.
	Major issues of encrypted control originate from the fact that
	such direct operations
	are in most cases restricted to addition and multiplication over integers.
In fact, it is well known that for a linear dynamic system to be realized over encrypted data, every element of its state matrix \emph{needs to be an integer} \cite{CDC18,ARC}.
For a SISO system, this is equivalent to the aforementioned constraint of having integer coefficients in the denominator of the transfer function.
Indeed, existing implementations of encrypted dynamic controllers without integer state matrices
heavily rely on additional resources,
such as
extra communication between the plant and the controller \cite{Kogiso,ElGamal,LQG,HEDD},
periodic reset of the controller \cite{reset},
and the bootstrapping technique \cite{NECSYS,bootstrap} which accompanies massive computational burden.

For this reason,
	the problem of finding a controller having an integer state matrix under performance guarantee has been extensively studied \cite{KimTAC23,KaoruTCNS,TSMC,CDC21,FIRapprox,Tavazoei22,Tavazoei23,Tavazoei23LCSS,Tava2026,CDC23}.
However, many of the previous results end up providing sufficient conditions that can only be satisfied by a limited class of plants
\cite{FIRapprox, Tavazoei22, Tavazoei23, Tavazoei23LCSS,CDC23,Tava2026}.
Such conditions include the strong stabilizability \cite{FIRapprox}, namely the existence of a stable controller that stabilizes the plant,
a property inherent to only a restricted class of systems \cite{ss}.
The methods in \cite{Tavazoei22,Tavazoei23} assume that certain algebraic integers exist based on
the plant's pole-zero configuration.
The approach in \cite{Tavazoei23LCSS}
involves a linear programming problem, yet is applicable if its solution meets a technical condition.
Even the recent work \cite{Tava2026}
merely provides
sufficient conditions
for the existence of a stabilizing controller
with integer coefficients.
While there exist methods applicable to a general class of plants,
their controllers
require additional input channels from the plant \cite{KimTAC23,KaoruTCNS,TSMC} or a time-varying implementation \cite{CDC21}.
To the best of our knowledge,
a generic stabilization method using a 
linear time-invariant (LTI) controller having integer coefficients
has not been established without additional assumptions.

In this paper, we show that
there always exists a controller consisting of integer coefficients that stabilizes a given LTI plant
without any assumption,
along with a constructive algorithm to find one.
Our approach is to first obtain a stabilizing controller, which does not have integer coefficients in general, and iteratively update this controller by increasing its order
so that it eventually consists of integer coefficients.
We provide an explicit upper bound on the number of these iterations.

Furthermore, the proposed method is applied to the \emph{conversion problem} \cite{KimTAC23,CDC21,FIRapprox,Tavazoei22,CDC23,KaoruTCNS,TSMC} stated as follows: Given a pre-designed controller, find an alternate controller having an integer state matrix that preserves the performance of the pre-designed controller.
Specifically, we solve the problem formulated in \cite{CDC23}, which aims to exactly preserve the transfer function from the reference signal to the plant output with respect to the closed-loop system.
Unlike in \cite{CDC23}, where this problem has
been solved for a restricted class of plants,
it is demonstrated that the principle of our method can be used to solve this conversion problem in general.

The rest of this paper is organized as follows.
Section~\ref{sec:pre} provides preliminaries and formulates the problem.
Section~\ref{sec:stabil} presents our main result,
a method to design a stabilizing controller having integer coefficients,
along with a numerical example.
In Section~\ref{sec:conv}, we address the conversion problem.
Finally, Section~\ref{sec:conclude} concludes the paper.

\textit{Notation:} The sets of integers, positive integers, and real numbers are denoted by $\Z$, $\N$, and $\R$, respectively.
The degree of a polynomial $a(z)$ is denoted by $\deg(a(z))$.
Define $\col\lbrace a_i\rbrace_{i=1}^n:=\left[a_n,\,\ldots,\,a_1\right]^\top$ for scalars $\lbrace a_i\rbrace_{i=1}^n$.
For $n\in\N$ and a polynomial $a(z)$
such that $\deg(a(z))\le n$,
define
\begin{equation}\label{eq:Toeplitz}
\calT_n(a(z)):=\scalebox{0.9}{$\begin{bmatrix}
		a_n & 0 & \cdots &  0\\
		a_{n-1} & a_n & \ddots & \vdots \\
		\vdots & \vdots &  \ddots & 0\\
		a_1 & a_2 & & a_n\\
		a_0 & a_1 &  &   a_{n-1}\\
		0 & a_0 & &  a_{n-2}\\
		\vdots & \ddots & \ddots & \vdots \\
		0 &\cdots & 0 & a_0
	\end{bmatrix}$}\in \R^{2n\times n},
	\end{equation}
	where $a_i$'s for $i=0,\,\ldots,\,\deg(a(z))$ are the $i$-th coefficients of $a(z)$ and $a_i=0$ for $i>\deg(a(z))$.
	We abuse notation and refer to \eqref{eq:Toeplitz} occasionally as $\calT_n(a)$, where $a=\col\lbrace a_i \rbrace_{i=0}^{n}\in\R^{n+1}$.
	The open ball of radius $r>0$ centered at $x\in\R^n$ according to the infinity norm is denoted by $B_r(x):=\lbrace y\in\R^n: \lVert y-x\rVert_\infty<r\rbrace$.
    Let $c^\ast$ denote the complex conjugate of a complex number $c$, and
	the zero vector of length $n$ is denoted by $0_n$.
	Let $\lVert \cdot \rVert$ denote the (induced) $1$-norm of a vector or a matrix. The ceiling and rounding operations are denoted by $\lceil \cdot \rceil$ and $\lceil \cdot \rfloor$, respectively.
	
	\section{Preliminaries and Problem Formulation}\label{sec:pre}
	
	\subsection{Preliminaries}\label{subsec:pre}
	
	We first review a method to solve the classical stabilization problem via
Bézout's identity,
and generalize the process into a mapping that is used throughout the paper.

It is well known that given a plant \eqref{eq:plant}, one can arbitrarily assign the closed-loop poles, i.e., the roots of the polynomial \eqref{eq:cl}, by designing a controller \eqref{eq:ctr}.
In terms of polynomials,
the following specifically holds:
Let coprime polynomials $D_p(z)$ and $N_p(z)$ be given and $n:=\deg(D_p(z))$.
Then,
for any polynomial $\gamma(z)$ of degree $2n$,
there exist polynomials $D_c(z)$ and $N_c(z)$
such that the polynomial \eqref{eq:cl} is equal to $\gamma(z)$ and $\deg(D_c(z))> \deg(N_c(z))$, where the latter ensures causality.
This can be shown as follows:
Since $D_p(z)$ and $N_p(z)$ are coprime, by Bézout's identity,
there exists a unique pair of $u(z)$ and $v(z)$ such that
\begin{equation}\label{eq:bezout}
u(z)D_p(z)+v(z)N_p(z)=1,
\end{equation}
$\deg(u(z))<\deg(N_p(z))$, and $\deg(v(z))<\deg(D_p(z))$.
By multiplying $\gamma(z)$ to both sides of \eqref{eq:bezout}, we obtain
\begin{multline*}
\underbrace{\left(u(z)\gamma(z)+d(z)N_p(z)\right)}_{=D_c(z)}D_p(z)\\+\underbrace{\left(v(z)\gamma(z)-d(z)D_p(z)\right)}_{=-N_c(z)}N_p(z)=\gamma(z),
\end{multline*}
where $d(z)$ is the quotient of $v(z)\gamma(z)$ divided by $D_p(z)$
and $-N_c(z)$ is the remainder.
Thus,
$\deg(N_c(z))<\deg(D_p(z))=n$, and hence $\deg(D_c(z))=n$.

This procedure can be interpreted as a way to generate a given polynomial from two coprime polynomials under some degree constraint.
Indeed, with a fixed $N_p(z)$, one is able to map a given pair of $D_p(z)$ and $\gamma(z)$ to $D_c(z)$.
The following definition generalizes such a mapping.

\begin{definition}\label{def:map}
Consider a polynomial $p(z)$, which is coprime to $N_p(z)$, and a polynomial $q(z)$, where $\deg(q(z))\geq \deg(p(z))$.
Then, the mapping $\calF$ is defined by $\calF:(p(z),q(z))\mapsto r(z)$,
where $r(z)$ satisfies
\begin{equation}\label{eq:map}
	p(z)r(z)+s(z)N_p(z)=q(z)
\end{equation}
for some polynomial $s(z)$ such that $\deg(s(z))<\deg(p(z))$.
\end{definition}

Note that the mapping $\calF$ is well-defined in the sense that the polynomial $r(z)$ of Definition~\ref{def:map} is uniquely determined, since $p(z)$ and $N_p(z)$ are coprime.

\subsection{Problem formulation}

The objective is to find a controller \eqref{eq:ctr} having integer coefficients that stabilizes the given plant \eqref{eq:plant}.
Namely, we find polynomials $D_c(z)$ and $N_c(z)$ such that \eqref{eq:cl} is Schur stable, $D_c(z)$ is an integer monic polynomial, and $\deg(N_c(z))<\deg(D_c(z))$, under the assumption that the polynomials $D_p(z)$ and $N_p(z)$ of \eqref{eq:plant} are coprime.
Without loss of generality, let $D_p(z)$ of \eqref{eq:plant} be a monic polynomial of degree $n$.
Then, the problem can be rewritten as follows.

\begin{problem}\label{prob:stabilization}
Find polynomials $\alpha(z)$, $\beta(z)$, and $\gamma(z)$ such that
\begin{equation}\label{eq:stabilization}
	\alpha(z)D_p(z)+\beta(z)N_p(z)=\gamma(z)
\end{equation}
and satisfy the followings:
\begin{enumerate}
	\item[(S1)] $\alpha(z)$ is an integer monic polynomial.
	\item[(S2)] $\gamma(z)$ is a Schur stable monic polynomial.
	\item[(S3)] $\deg(\beta(z))<\deg(\alpha(z))$.
\end{enumerate}
\end{problem}

By solving Problem~\ref{prob:stabilization}, a controller with integer coefficients can be constructed as
$D_c(z)=\alpha(z)$ and $N_c(z)=-\beta(z)$,
ensuring that it is strictly proper by (S3)
and the closed-loop system is stable by (S2).
Note that neither the roots of $\gamma(z)$ (the poles of the closed-loop system) nor the degree of $\alpha(z)$ (the order of the controller) is designated in advance.

Without loss of generality, we assume that $N_p(0)\neq 0$ for the rest of this paper.
Otherwise,
$N_p(z)$ can be factorized as
$N_p(z)=z^l\tilde{N}_p(z)$, where
$\tilde{N}_p(0)\neq 0$ and $l\in\N$.
Then, given a solution $(\alpha(z),\beta(z),\gamma(z))$
to Problem~\ref{prob:stabilization}
with respect to $D_p(z)$ and $\tilde{N}_p(z)$,
$(z^l\alpha(z),\beta(z),z^l\gamma(z))$ is
a solution to Problem~\ref{prob:stabilization}
with respect to $D_p(z)$ and $N_p(z)$.

\begin{remark}
Given a strictly proper plant \eqref{eq:plant},
one can design a controller \eqref{eq:ctr} with integer coefficients that is not necessarily strictly proper as follows;
let $(\alpha(z),\beta(z),\gamma(z))$ be a solution to Problem~\ref{prob:stabilization}
with respect to $(D_p(z),\tilde{\beta}(z)N_p(z))$,
where $\tilde{\beta}(z)$ is an
arbitrary degree-$1$ polynomial coprime to $D_p(z)$.
Then,
$\deg(N_c(z))\leq \deg(D_c(z))$
with
$D_c(z)=\alpha(z)$ and $N_c(z)=-\tilde{\beta}(z)\beta(z)$.
\end{remark}

\section{Main Result}\label{sec:stabil}

Now we propose a method to find a stabilizing controller having integer coefficients by solving Problem~\ref{prob:stabilization}.
The challenge of Problem~\ref{prob:stabilization} lies in the integer constraint (S1), since
finding $(\alpha(z),\beta(z),\gamma(z))$ without this condition is a mere stabilization problem as illustrated in Section~\ref{subsec:pre}.
Thus, our approach is to first find $(\alpha(z),\beta(z),\gamma(z))$ satisfying all other conditions of Problem~\ref{prob:stabilization} except (S1), and iteratively update it so that $\alpha(z)$ eventually becomes an integer polynomial.

We begin with a simple observation;
let $(\alpha(z),\beta(z),\gamma(z))$ satisfying \eqref{eq:stabilization} be given.
Then, by multiplying some polynomial $r(z)$ to both sides of \eqref{eq:stabilization},
$(\alpha(z),\beta(z),\gamma(z))$ can be updated to $(\alpha^+(z),\beta^+(z),\gamma^+(z))$ as
\begin{multline}\label{eq:pm}
\underbrace{\left(r(z)\alpha(z)+w(z)N_p(z)\right)}_{=\alpha^+(z)}D_p(z)\\+\underbrace{\left(r(z)\beta(z)-w(z)D_p(z)\right)}_{=\beta^+(z)}N_p(z)=\underbrace{r(z)\gamma(z)}_{=\gamma^+(z)}
\end{multline}
using some polynomial $w(z)$, while satisfying \eqref{eq:stabilization}.
Additionally, if $r(z)$ is Schur stable and monic, then so is $\gamma^+(z)$ given that $\gamma(z)$ satisfies (S2).
In this manner, we iteratively update the polynomials by properly selecting $r(z)$ and $w(z)$.

In what follows,
our framework of updating these polynomials is described in detail,
and then a method to solve Problem~\ref{prob:stabilization} under this framework is proposed.

\subsection{Proposed framework}\label{subsec:update}

Throughout the iterations, let $\alpha(z)$ be in the form of
\begin{equation}\label{eq:alpha}
\alpha(z)=z^Na(z)=z^{N+n}+a_{n-1}z^{N+n-1}+\cdots+a_0z^N
\end{equation}
for some $N\ge 0$ and a monic polynomial $a(z)$ of degree $n$, whose $i$-th coefficient is denoted by $a_i$.
This enables us to consider only the $n$-coefficients of the higher order terms, except the leading one.

We begin by finding $(\alpha^\ini(z),\beta^\ini(z),\gamma^\ini(z))$ that satisfies every condition of Problem~\ref{prob:stabilization} other than (S1).
As described in Section~\ref{sec:pre},
this can be done by first selecting
$\gamma^\ini(z)$ as a Schur stable monic polynomial of degree $2n$, and then
letting
\begin{equation}\label{eq:alpha ini}
\alpha^\ini(z)=\calF(D_p(z),\gamma^\ini(z)).
\end{equation}
This determines $\beta^\ini(z)$ by \eqref{eq:stabilization}, which also satisfies (S3) by Definition~\ref{def:map}.
Note that $\alpha^\ini(z)$ has the form of \eqref{eq:alpha} with $N=0$, as its degree is $n$.

As mentioned earlier, we select $r(z)$, a Schur stable monic polynomial, at each iteration.
Let the degree of $r(z)$ be $n$, and suppose the \emph{current} $(\alpha(z),\beta(z),\gamma(z))$ satisfies \eqref{eq:stabilization}, (S2), (S3), and \eqref{eq:alpha}.
Then, the \emph{next} $(\alpha^+(z),\beta^+(z),\gamma^+(z))$, which is defined by \eqref{eq:pm} with some polynomial $w(z)$, meets \eqref{eq:stabilization} and (S2).
Moreover, if
\begin{equation}\label{eq:deg w}
\deg(w(z))<\deg(\alpha(z)),
\end{equation}
then $\alpha^+(z)$ and $\beta^+(z)$ satisfy (S3)\footnote{
This is because $\deg(w(z)N_p(z))<\deg(r(z)\alpha(z))=\deg(\alpha^+(z))$
and
$\deg(r(z)\alpha(z))$ is greater than both $\deg(r(z)\beta(z))$ and $\deg(w(z)D_p(z))$.}.

Now it remains to make $\alpha^+(z)$ be in the form of \eqref{eq:alpha}.
It follows from \eqref{eq:deg w} that
$\deg(\alpha^+(z))=\deg(\alpha(z))+n$.
Thus, having $\alpha(z)=z^Na(z)$,
the next $\alpha^+(z)$ should be expressed as $z^{N+n}a^+(z)$
for some degree-$n$ monic polynomial $a^+(z)$.
In other words,
we need $w(z)$ such that
\begin{equation}\label{eq:w}
z^N a(z)r(z)+w(z)N_p(z)=z^{N+n}a^+(z)
\end{equation}
for some degree-$n$ monic polynomial $a^+(z)$,
as well as \eqref{eq:deg w} holds.
In fact, it directly follows from Definition~\ref{def:map} that such $w(z)$ is unique, since
\begin{equation*}
a^+(z)=\calF(z^{N+n},z^Na(z)r(z)).
\end{equation*}
This indicates that given $r(z)$ and the current $\alpha(z)=z^Na(z)$, the next $\alpha^+(z)=z^{N+n}a^+(z)$ is uniquely determined.

To further investigate the relation between $a(z)$, $a^+(z)$, and $r(z)$,
we represent them as the
\emph{coefficient vectors}, namely
\begin{equation*}
\mathrm{a}:=\col\lbrace a_i\rbrace_{i=0}^{n-1},\quad \mathrm{a}^+:=\col\lbrace a^+_i\rbrace_{i=0}^{n-1},\quad \mathrm{r}:=\col\lbrace r_i \rbrace_{i=0}^{n-1},
\end{equation*}
where $a_i^+$ and $r_i$ are the $i$-th coefficients of $a^+(z)$ and $r(z)$, respectively.
For analysis, we also define $\Delta(x)$ for $x\in\R^n$ by
\begin{equation*}
\Delta(x):=\Delta_1(x)-P_1P_2^{-1}\Delta_2(x),
\end{equation*}
where $P_1$, $P_2$, $\Delta_1(x)$, and $\Delta_2(x)$ are $n\times n$ matrices defined by
\begin{equation*}
\begin{bmatrix}
P_1\\ P_2
\end{bmatrix}:=\calT_n(N_p(z)),\quad
\begin{bmatrix}
\Delta_1(x)\\ \Delta_2(x)
\end{bmatrix}:=\calT_n\left(\begin{bmatrix}
1\\ x
\end{bmatrix}\right).
\end{equation*}
Note that $P_2$ is nonsingular due to $N_p(0)\neq 0$ because it is an upper triangular matrix whose main-diagonal elements are the constant term of $N_p(z)$.
With these definitions in hand, we provide the following proposition.

\begin{proposition}\label{prop:a update}
Let $a(z)$ and $r(z)$ be monic polynomials of degree $n$, and $N\ge 0$.
Then, a polynomial $a^+(z)$ satisfying \eqref{eq:w} with some polynomial $w(z)$ of degree less than $N+n$ is uniquely determined as $\mathrm{a}^+=\mathrm{a}+\Delta(\mathrm{a})\mathrm{r}$.
\end{proposition}

\begin{proof}
Since $N_p(0)\neq 0$, $z^N$ is a divisor of $w(z)$ by \eqref{eq:w}. By letting $w(z)=z^Nb(z)$, it follows from \eqref{eq:w} that
\begin{equation}\label{eq:aplus}
a(z)r(z)+b(z)N_p(z)=z^na^+(z).
\end{equation}
We rewrite \eqref{eq:aplus} as
\begin{multline}\label{eq:mtx}
\scalebox{0.88}{$\begin{bmatrix}
		1 & 0 & \cdots & 0\\
		a_{n-1} & 1 & & \\
		\vdots & \vdots &\ddots & \\
		a_0 & a_1 & & 1\\
		\hline
		0 & a_0 & & a_{n-1}\\
		\vdots  & \ddots & \ddots& \vdots\\
		0 & \cdots & 0 & a_0
	\end{bmatrix}\begin{bmatrix}
		1\\ r_{n-1}\\ \vdots \\ r_0
	\end{bmatrix}$}
+\scalebox{0.88}{$\begin{bmatrix}
		0 & \cdots & 0\\
		0 & \cdots & 0\\
		p_{n-1} & & \\
		\vdots & \ddots & \\
		\hline
		p_0 & & p_{n-1}\\
		& \ddots & \vdots \\
		& & p_0
	\end{bmatrix}\begin{bmatrix}
		b_{n-1}\\ \vdots \\ b_0
	\end{bmatrix}$}\\
=\left[\begin{array}{cccc|ccc}
	1 & a^+_{n-1} & \cdots & a^+_0 & 0 & \cdots & 0
\end{array}\right]^\top,
\end{multline}
where
$b_i$ is the $i$-th coefficient of $b(z)$ for $i\leq\deg(b(z))$ and $b_i=0$ otherwise.
Similarly, $p_i$ denotes the $i$-th coefficient of $N_p(z)$ when $i\leq \deg(N_p(z))$, and otherwise, $p_i=0$.
Let $\mathrm{b}:=\col\lbrace b_i\rbrace_{i=0}^{n-1}$.
Then, \eqref{eq:mtx} is rewritten as
\begin{equation*}
\begin{bmatrix}
	1 & 0_n^\top\\
	\mathrm{a} & \Delta_1(\mathrm{a})\\
	0_n & \Delta_2(\mathrm{a})
\end{bmatrix}\begin{bmatrix}
	1\\ \mathrm{r}
\end{bmatrix}+\begin{bmatrix}
	0_n^\top\\ P_1\\ P_2
\end{bmatrix}\mathrm{b}=\begin{bmatrix}
	1\\ \mathrm{a}^+\\ 0_n
\end{bmatrix}.
\end{equation*}
As $\mathrm{b}=-P_2^{-1}\Delta_2(\mathrm{a})\mathrm{r}$ and $\mathrm{a}^+=\mathrm{a}+\Delta_1(\mathrm{a})\mathrm{r}+P_1\mathrm{b}$, the proof is concluded.
\end{proof}

With Proposition~\ref{prop:a update}, we are now able to interpret our update framework as the following dynamic system;
\begin{equation}\label{eq:dyn}
x_{k+1}=x_k+\Delta(x_k)u_k,
\end{equation}
where $x_k\in\R^n$ is the state and $u_k\in\R^n$ is the input.
Here, the initial state $x_0$ is determined by $\alpha^\ini(z)$.
We can regard the state space $\R^n$ as the space of monic polynomials with degree $n$, where a vector $v=\col\lbrace v_i\rbrace_{i=0}^{n-1}\in\R^n$ corresponds to the polynomial
\begin{equation*}
p_v(z):=z^n+v_{n-1}z^{n-1}+\cdots+v_1z+v_0,
\end{equation*}
and vice versa.
Thus, $p_{x_k}(z)=a(z)$ and $p_{u_k}(z)=r(z)$ at the $k$-th iteration.

Now our goal is to design the input $u_k$ to the system \eqref{eq:dyn} such that $p_{u_k}(z)$ is Schur stable at each $k$---recall that $r(z)$ should be Schur stable---and the state $x_k$ reaches some integer vector $x^\star\in\Z^n$ within a finite number of time steps.
Observe from \eqref{eq:dyn} that any point in $\R^n$ is reachable from $x_0$ if $\Delta(x_0)$ is invertible.
The following lemma gives a necessary and sufficient condition for the invertibility of $\Delta(x)$.

\begin{lemma}\label{lem:invert}
For any $x\in\R^n$, $\Delta(x)$ is invertible if and only if $p_x(z)$ and $N_p(z)$ are coprime.
\end{lemma}

\begin{proof}
Since $\Delta(x)$ is the Schur complement of the block $P_2$ of
\begin{equation*}
\Gamma(x):=\begin{bmatrix}
	\Delta_1(x) & P_1\\
	\Delta_2(x) & P_2
\end{bmatrix}\in\R^{2n\times 2n},
\end{equation*}
$\Delta(x)$ is invertible if and only if $\Gamma(x)$ is invertible.
It can be shown that there exists a nonzero vector $[u^\top,\,v^\top]^\top\in\R^{2n}$, where
$u\in\R^n$ and $v\in\R^n$, such that
\begin{equation*}
\Gamma(x)\begin{bmatrix}
	u\\ v
\end{bmatrix}=\begin{bmatrix}
	\calT_n\left(\begin{bmatrix}
		1\\ x
	\end{bmatrix}\right) & \calT_n(N_p(z))
\end{bmatrix}\begin{bmatrix}
	u\\ v
\end{bmatrix}=0_{2n},
\end{equation*}
if and only if
there exist nonzero polynomials $u(z)$ and $v(z)$ such that
\begin{equation*}
u(z)p_x(z)+v(z)N_p(z)=0,
\end{equation*}
$\deg(u(z))<n$, and $\deg(v(z))<n$.
We show that this is equivalent to $p_x(z)$ and $N_p(z)$ not being coprime.
If $p_x(z)$ and $N_p(z)$ are not coprime, then the existence of such $u(z)$ and $v(z)$ is trivial.
Conversely, if $p_x(z)$ and $N_p(z)$ are coprime, then $p_x(z)$ should divide $v(z)$, which contradicts the fact that $\deg(p_x(z))>\deg(v(z))$.
Therefore, $\Gamma(x)$ is invertible if and only if $p_x(z)$ and $N_p(z)$ are coprime.
\end{proof}

It is easily verified that if we select $\gamma^\ini(z)$ that is coprime to $N_p(z)$, then $p_{x_0}(z)=\alpha^\ini(z)$ is also coprime to $N_p(z)$.
Thus, with $u_0=\Delta(x_0)^{-1}(x^\star-x_0)$, the state reaches any $x^\star\in\Z^n$ in one time step.
However, this may lead to $p_{u_0}(z)$ that is not Schur stable.

In this regard,
we use the fact that $p_u(z)$ is Schur stable for any $u\in\R^n$ if $\lVert u\rVert<1$ by Rouché's theorem\footnote{Consider a monic polynomial $p(z)=z^n+v_{n-1}z^{n-1}+\cdots+v_0$. If $\lvert v_{n-1}z^{n-1}+\cdots+v_0\rvert<\lvert z^n\rvert$ for $\lvert z\rvert=1$, then $z^n$ and $p(z)$ have the same number of zeros inside the unit circle.} \cite{complex}.
Therefore, in the next subsection,
we move $x_k$ gradually to an integer vector $x^\star$ within the area where $p_x(z)$ is coprime to $N_p(z)$, using a bounded input.

\subsection{Solution to Problem~\ref{prob:stabilization}}\label{subsec:dyn}

We first specify a domain in which the state $x_k$ of \eqref{eq:dyn} can evolve while keeping $\Delta(x_k)$ invertible,
and then show that there exists an integer vector $x^\star$, which is our destination, inside this domain.
Accordingly, the input $u_k$ is designed so that the state reaches $x^\star$ after a finite number of time steps.

To this end, we factorize $N_p(z)$ as
\begin{equation*}
N_p(z)=c\prod_{j=1}^{n_r}\left(z-\lambda_j\right)\prod_{j=1}^{n_c}\left(z-\eta_j\right)\left(z-\eta^\ast_j\right),
\end{equation*}
where $c\in\R$,
$\lambda_j$'s are the real roots, and $\eta_j$'s are the complex non-real roots.
Then, for $j=1,\,2,\,\ldots,\,n_r$,
the set of degree-$n$ monic polynomials having $\lambda_j$ as their root corresponds to the hyperplane
\begin{equation}\label{eq:real}
\left\lbrace x\in\R^n: p_x(\lambda_j)=0\right\rbrace=\left\lbrace x\in\R^n: \phi_j^\top x=\psi_j\right\rbrace,
\end{equation}
where $\phi_j\in\R^n$ and $\psi_j\in\R$ are defined by
\begin{equation*}
\left[-\psi_j,\,\phi_j^\top\right]:=\left[\lambda_j^n,\,\lambda_j^{n-1},\,\ldots,\,1\right].
\end{equation*}

Similarly, for $j=1,\,2,\,\ldots,\,n_c$,
the set of $x\in\R^n$ such that $p_x(z)$ has both $\eta_j$ and $\eta_j^\ast$ as its roots is
\begin{equation*}
\left\lbrace x\in\R^n: \phi_{n_r+2j-1}^\top x=\psi_{n_r+2j-1}\,\,\text{and}\,\,\phi_{n_r+2j}^\top x=\psi_{n_r+2j}\right\rbrace
\end{equation*}
where
\begin{multline*}
\begin{bmatrix}
-\psi_{n_r+2j-1} & \phi_{n_r+2j-1}^\top\\
-\psi_{n_r+2j} & \phi_{n_r+2j}^\top
\end{bmatrix}\\:=\frac{1}{2}\begin{bmatrix}
1 & 1\\ -i & i
\end{bmatrix}\begin{bmatrix}
\eta_j^n & \eta_j^{n-1} & \cdots & 1\\
\left(\eta_j^\ast\right)^n & \left(\eta_j^\ast\right)^{n-1} & \cdots & 1
\end{bmatrix}\in\R^{2\times (n+1)}
\end{multline*}
with $i^2=-1$.
By definition, the real and the imaginary parts of $p_x(\eta_j)$ are $\phi_{n_r+2j-1}^\top x-\psi_{n_r+2j-1}$ and $\phi_{n_r+2j}^\top x-\psi_{n_r+2j}$, respectively.

\begin{figure}
\centering
\includegraphics[width=0.28\textwidth]{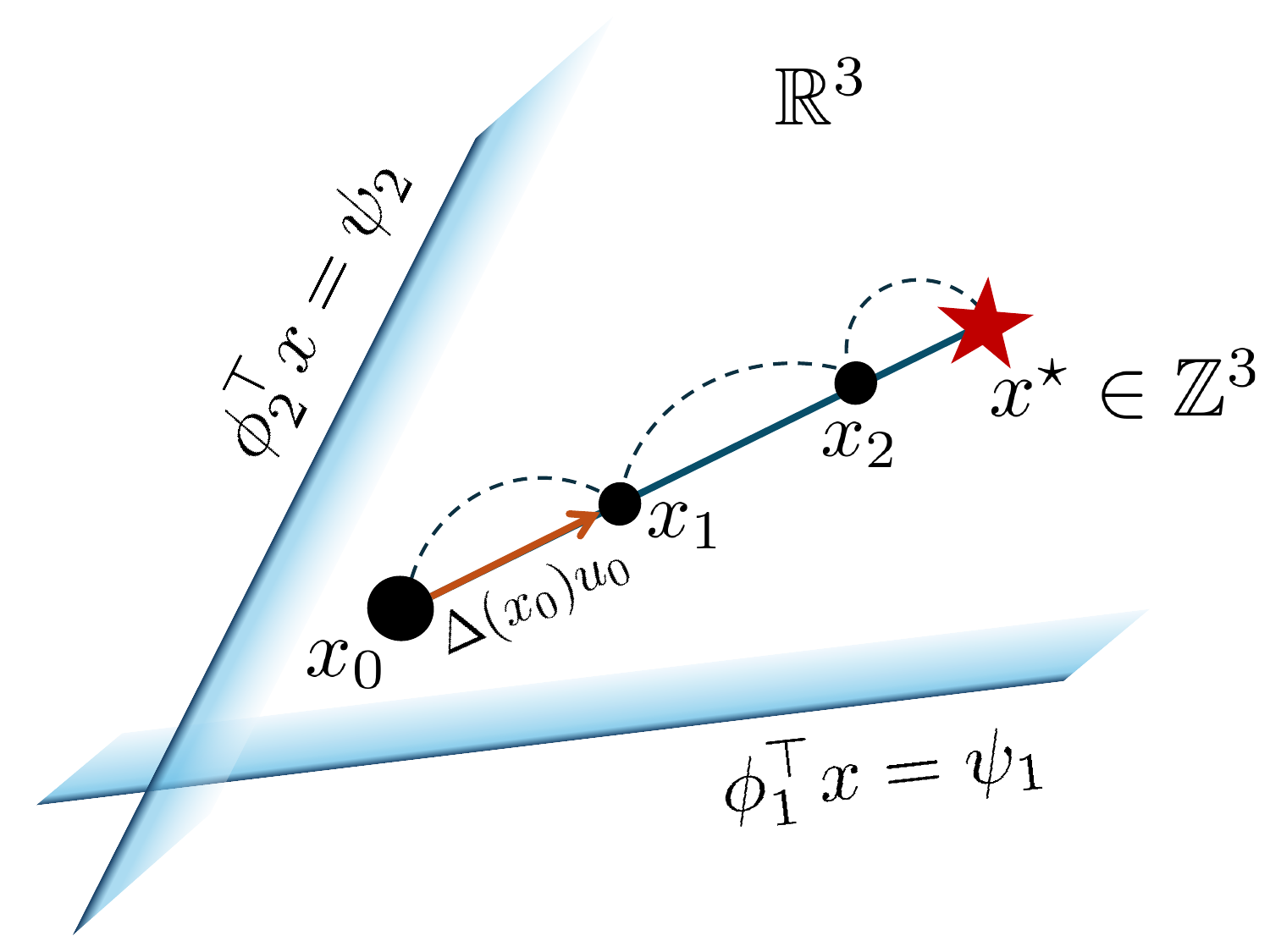}
\caption{
Illustration of the state $x_k$, the destination $x^\star$, and the hyperplanes \eqref{eq:hyperplane}
when $n=3$, $n_r=2$, and $n_c=0$.
}
\label{fig}
\end{figure}

Therefore, by Lemma~\ref{lem:invert},
as long as the state $x_k$ avoids the hyperplanes
\begin{equation}\label{eq:hyperplane}
\left\lbrace x\in\R^n: \phi_t^\top x=\psi_t\right\rbrace\,\,\,\text{for}\,\,\, t=1,\,2,\,\ldots,\,n_r+2n_c,
\end{equation}
$\Delta(x_k)$ is invertible.
We aim to find $x^\star\in\Z^n$ such that the line segment
\begin{equation*}
\Omega:=\left\lbrace \rho x_0+(1-\rho)x^\star: \rho\in[0,1]\right\rbrace.
\end{equation*}
does not \emph{cross} these hyperplanes, so that $\Delta(x)$ is invertible for any $x\in\Omega$, as depicted in Fig.~\ref{fig}.
In other words, $x^\star$ should be on the \emph{same side} of each hyperplane as the initial state $x_0$.
This condition is equivalent to
\begin{equation}\label{eq:xstar}
\left( \phi_t^\top x_0-\psi_t\right) \left( \phi_t^\top x^\star-\psi_t\right)>0\quad \forall t\in\calI,
\end{equation}
where $\calI:=\lbrace t: \phi_t^\top x_0-\psi_t\neq 0\rbrace$.

We introduce $\calI$ to account for the fact that
even when $p_{x_0}(z)$ and $N_p(z)$ are coprime,
one of the real and the imaginary parts of $p_{x_0}(\eta_j)$ can be zero for some $j=1,\,2,\,\ldots,\,n_c$,
and hence
there may exist $t\in (n_r,n_r+2n_c]$ such that $\phi_t^\top x_0-\psi_t=0$.
Nevertheless,
the existence of $x^\star\in\Z^n$ satisfying \eqref{eq:xstar} is ensured by the following proposition.

\begin{proposition}\label{prop:xstar}
Given $x_0\in\R^n$ such that $p_{x_0}(z)$ and $N_p(z)$ are coprime, there exists $x^\star\in\Z^n$ satisfying \eqref{eq:xstar}.
\end{proposition}

\begin{proof}
Define $\Phi\subset \R^n$ as
\begin{equation}\label{eq:Phi}
\Phi:=\lbrace x\in\R^n: (\phi_t^\top x_0-\psi_t)(\phi_t^\top x-\psi_t)>0\quad\forall t\in\calI\rbrace.
\end{equation}
We show that there exists a ball of radius $1$ according to the infinity norm, which always contains an integer vector, inside $\Phi$.
Since $\Phi$ is open, there exists $\delta>0$ such that $B_\delta(x_0)\subset \Phi$.
Consider $v\in\R^n$ such that $p_v(z)=z^{n-n_r-2n_c}N_p(z)/c$.
We show that $\mathcal{B}_1:=B_1((x_0-v)/\delta+v)\subset \Phi$.
For any $y\in \mathcal{B}_1$,
it can be verified that $\delta(y-v)+v\in B_\delta(x_0)\subset \Phi$.
Then,
\begin{multline*}
\left(\phi_t^\top x_0-\psi_t\right)\left(\phi_t^\top\left(\delta(y-v)+v\right)-\psi_t\right)\\
=\delta\left(\phi_t^\top x_0-\psi_t\right)\left(\phi_t^\top y-\psi_t\right)>0\quad \forall t\in\calI,
\end{multline*}
since $\phi_t^\top v=\psi_t$ for all $t\in \calI$ by the definition of $v$.
Therefore, $y\in \Phi$, which concludes the proof.
\end{proof}

Seeing Fig.~\ref{fig}, one may think of a case when the hyperplanes are parallel and close to each other, in a way that $x^\star\in\Z^n$ does not exist.
However, the intersection of the hyperplanes \eqref{eq:hyperplane} is not empty, since there always exists a monic polynomial of degree $n$ having every root of $N_p(z)$ as its root.
This plays a key role in the proof of Proposition~\ref{prop:xstar}.

Having the destination $x^\star\in\Z^n$, we design the input as
\begin{equation}\label{eq:u}
u_k=\scalebox{0.95}{$\displaystyle\begin{cases}
	\Delta(x_k)^{-1}\left(x^\star-x_k \right), & \text{if}\,\left\lVert \Delta(x_k)^{-1}\left(x^\star-x_k \right)\right\rVert<1,\\
	\displaystyle
	\frac{\mu\Delta(x_k)^{-1}\left(x^\star-x_k \right)}{\left\lVert\Delta(x_k)^{-1}\left(x^\star-x_k \right)\right\rVert}, & \text{otherwise,}
\end{cases}$}
\end{equation}
with some $\mu\in(0,1)$.
By design, $\lVert u_k\rVert<1$ for all $k\ge 0$.
Recall that this ensures $p_{u_k}(z)$ to be Schur stable.

As depicted in Fig.~\ref{fig}, the input \eqref{eq:u} drives the state towards $x^\star$ along the line segment $\Omega$.
Since $x^\star$ is located so that $\Delta(x)$ is invertible for any $x\in\Omega$,
the input \eqref{eq:u} is well-defined.
Moreover, it is guaranteed that the state reaches $x^\star$ in a finite number of time steps,
as stated in the following proposition.

\begin{proposition}\label{prop:finite}
Suppose that the system \eqref{eq:dyn} has an initial state $x_0\in\R^n$ such that $p_{x_0}(z)$ and $N_p(z)$ are coprime.
Then, given $x^\star\in\Z^n$ satisfying \eqref{eq:xstar},
the input \eqref{eq:u} is well-defined for all $k\ge 0$ and achieves
$x_T=x^\star$ for some $T\ge 0$ such that
\begin{equation*}
	T\leq \left\lceil \frac{\sigma}{\mu}\left\lVert x^\star-x_0\right\rVert\right\rceil=:\bar{T},
\end{equation*}
where $\sigma:=\max_{x\in\Omega}\lVert \Delta^{-1}(x)\rVert$.
\end{proposition}

\begin{proof}
Since $p_{x_0}(z)$ and $N_p(z)$ are coprime, $\lbrace 1,\,\ldots,\,n_r\rbrace\subset\calI$ by \eqref{eq:real}.
For $j=1,\,\ldots,\,n_c$, at most one of $n_r+2j-1$ and $n_r+2j$ is in $\calI$.
Then, for any $x\in\Phi$, where $\Phi$ is defined by \eqref{eq:Phi}, $p_x(z)$ is coprime to $N_p(z)$.
As $\Phi$ is convex, $\Omega\subset \Phi$, and hence $\Delta(x)$ is invertible for all $x\in\Omega$ by Lemma~\ref{lem:invert}.
Thus, $0<\sigma<\infty$ since $\lVert \Delta^{-1}(x)\rVert$ is a well-defined continuous function on $\Omega$, which is compact.
In addition, it can be verified by induction that $x_k\in\Omega$ and the input \eqref{eq:u} is well-defined for all $k\ge 0$.
If there exists $k\in[0,\bar{T})$ such that $\lVert \Delta(x_k)^{-1}(x^\star-x_k)\rVert < 1$, then $x_{k+1}=x^\star$ by \eqref{eq:u} and the proof ends.
Suppose otherwise that $\lVert \Delta(x_k)^{-1}(x^\star-x_k)\rVert \geq 1$ for all $k\in[0,\bar{T})$.
Then,
\begin{equation}\label{eq:proppf}
	x_{k+1}=x_k+\frac{\mu}{\left\lVert \Delta(x_k)^{-1}\left(x^\star-x_k\right)\right\rVert}\left(x^\star-x_k\right),
\end{equation}
and thus
\begin{equation}\label{eq:proppf2}
	x_{k+1}\in \left\lbrace \rho x_k+(1-\rho)x^\star: \rho\in (0,1) \right\rbrace
\end{equation}
for all $k\in[0,\bar{T})$,
since $\mu/\lVert \Delta(x_k)^{-1}(x^\star-x_k)\rVert\in (0,1)$.
It follows from \eqref{eq:proppf} that
\begin{equation*}
	\left\lVert x_{k+1}-x_k\right\rVert=\frac{\mu\left\lVert x^\star-x_k\right\rVert}{\left\lVert \Delta(x_k)^{-1}\left(x^\star-x_k\right)\right\rVert}\geq \frac{\mu}{\sigma}
\end{equation*}
for all $k\in [0,\bar{T})$.
Since \eqref{eq:proppf2} implies that $x_1,\,x_2,\,\ldots,\,x_{\bar{T}}$ are located sequentially in line between $x_0$ and $x^\star$,
\begin{equation*}
	\left\lVert x_{\bar{T}}-x_0\right\rVert\geq \frac{\mu}{\sigma}\left\lceil \frac{\sigma}{\mu}\left\lVert x^\star-x_0\right\rVert\right\rceil\geq \left\lVert x^\star-x_0\right\rVert,
\end{equation*}
which contradicts the fact that $x_{\bar{T}}\in\Omega$ and $x_{\bar{T}}\neq x^\star$ by \eqref{eq:proppf2}.
This concludes the proof.
\end{proof}

\subsection{Overall procedure and main theorem}\label{subsec:main}

\begin{algorithm}
\caption{Solving Problem~\ref{prob:stabilization}.}\label{alg:stabil}
\begin{algorithmic}[1]
	\renewcommand{\algorithmicrequire}{\textbf{Input:}}
	\Require $D_p(z),\,N_p(z)$.
	\State $n\gets \deg(D_p(z))$.
	\State Choose a Schur stable monic polynomial $\gamma^\ini(z)$ of degree $2n$ that is coprime to $N_p(z)$. $\gamma(z)\gets \gamma^\ini(z)$.
	\State Perform \eqref{eq:alpha ini}. $N\gets 0$.
	\State
	Let $x_0\in\R^n$ such that $p_{x_0}(z)=\alpha^\ini(z)$. $k\gets 0$.
	\State Find $x^\star\in\Z^n$ satisfying \eqref{eq:xstar}. Select $\mu\in (0,1)$.
	\While{$x_k\neq x^\star$}\label{step:while start}
	\State Perform \eqref{eq:u}.
	\State Perform \eqref{eq:dyn}.
	\State $\gamma(z)\gets p_{u_k}(z)\gamma(z)$, $N\gets N+n$, $k\gets k+1$.
	\EndWhile\label{step:while end}
	\State $\alpha(z)\gets z^Np_{x^\star}(z)$.
	\State $\beta(z)\gets (\gamma(z)-\alpha(z)D_p(z))/N_p(z)$.
	\renewcommand{\algorithmicrequire}{\textbf{Output:}}
	\Require $\alpha(z),\,\beta(z),\,\gamma(z)$.
\end{algorithmic}
\end{algorithm}

In summary,
we start from $(\alpha^\ini(z),\beta^\ini(z),\gamma^\ini(z))$ that satisfies all other conditions of Problem~\ref{prob:stabilization} except (S1),
and iteratively update $(\alpha(z),\beta(z),\gamma(z))$ as in \eqref{eq:pm} so that it eventually achieves (S1).
We emphasize that $\alpha(z)$ corresponds to the state of \eqref{eq:dyn}, and $r(z)$, which we select at each iteration, is determined by the input \eqref{eq:u}.
Proposition~\ref{prop:finite} guarantees that $\alpha(z)$ becomes an integer polynomial after a finite number of iterations, so that Problem~\ref{prob:stabilization} is solved.
The overall procedure is presented as Algorithm~\ref{alg:stabil}.

From what we have discussed,
it is derived that Algorithm~\ref{alg:stabil} returns a solution to Problem~\ref{prob:stabilization} for any given plant \eqref{eq:plant}, and thus a stabilizing controller with integer coefficients can be found; this is stated in the following theorem.

\begin{theorem}\label{thm:stabil}
Given a plant \eqref{eq:plant},
there exists a controller \eqref{eq:ctr}
such that
$D_c(z)$ is an integer monic polynomial and
the closed-loop system is stable.
Furthermore, such a controller can be designed from the outputs of Algorithm~\ref{alg:stabil} as $D_c(z)=\alpha(z)$ and $N_c(z)=-\beta(z)$.
\end{theorem}

\begin{proof}
It suffices to show that Algorithm~\ref{alg:stabil} returns a solution to Problem~\ref{prob:stabilization}.
As $\gamma^\ini(z)$ is coprime to $N_p(z)$, so is $\alpha^\ini(z)$, and hence the assumption of Proposition~\ref{prop:finite} holds.
By Proposition~\ref{prop:finite},
Steps~6--10 of Algorithm~\ref{alg:stabil} are repeated only a finite number of times, achieving $x_k=x^\star$.
By construction, (S1) and (S2) of Problem~\ref{prob:stabilization} hold.

We prove by induction that after each iteration,
\begin{equation}\label{eq:thmpf}
	p_{x_k}(z)=\calF(z^ND_p(z),\gamma(z)),
\end{equation}
since this implies by Definition~\ref{def:map} that \eqref{eq:stabilization} and (S3) are satisfied after Step~12.
Indeed, \eqref{eq:thmpf} holds when $k=0$ and $\gamma(z)=\gamma^\ini(z)$ by \eqref{eq:alpha ini} and Step~4.
We show that if \eqref{eq:thmpf} holds, then
\begin{equation*}
	p_{x_{k+1}}(z)=\calF(z^{N+n}D_p(z),p_{u_k}(z)\gamma(z)),
\end{equation*}
i.e., \eqref{eq:thmpf} written with respect to the next iteration.
From \eqref{eq:thmpf},
\begin{equation}\label{eq:thmpf2}
	z^ND_p(z)p_{x_k}(z)p_{u_k}(z)+\eta(z)N_p(z)=p_{u_k}(z)\gamma(z),
\end{equation}
where $\eta(z)$ is a polynomial of degree less than $N+2n$.
By Proposition~\ref{prop:a update},
\begin{align*}
	p_{x_{k+1}}(z)&=\calF(z^{N+n},z^Np_{x_k}(z)p_{u_k}(z))\\
	&=\calF(z^{N+n}D_p(z),z^ND_p(z)p_{x_k}(z)p_{u_k}(z))\\
	&=\calF(z^{N+n}D_p(z),p_{u_k}(z)\gamma(z)-\eta(z)N_p(z))\\
	&=\calF(z^{N+n}D_p(z),p_{u_k}(z)\gamma(z)),
\end{align*}
where
the second and the last equality follow directly from Definition~\ref{def:map} and the third equality comes from \eqref{eq:thmpf2}.
\end{proof}

\begin{remark}
The degree of $\alpha(z)$ from Algorithm~\ref{alg:stabil} is less than or equal to
$n\bar{T}+n$, which is determined by the choice of $x_0$, $x^\star$, and $\mu$.
Note that the initial state $x_0$ is determined by the choice of $\gamma^\ini(z)$.
In practice, the destination point $x^\star$ can be found by simply investigating integer vectors nearby $x_0$, with checking if the condition \eqref{eq:xstar} holds.
\end{remark}

\begin{remark}
If the condition \eqref{eq:xstar} holds for $x^\star=0_n$,
then the final $\alpha(z)$ becomes a monomial, and thus every pole of the resulting controller is at the origin. This implies that the plant \eqref{eq:plant} is strongly stabilizable \cite{ss}, i.e., stabilizable by a stable controller.
In fact, a controller with integer coefficients is stable only when all of its poles are at the origin \cite{stableCtr}.
\end{remark}

\subsection{Numerical example}\label{subsec:stabil ex}

This subsection provides an illustration of applying Algorithm~\ref{alg:stabil} to a linearized inverted pendulum \cite{franklin}, written by
\begin{equation}\label{eq:invpen}
\begin{aligned}
	\left(I+ml^2\right)\ddot{\phi}(t)-mgl\phi(t)&=ml\ddot{x}(t),\\
	\left(M+m\right)\ddot{x}(t)+b\dot{x}(t)-ml\ddot{\phi}(t)&=u(t),\,\,
	y(t)=x(t),
\end{aligned}
\end{equation}
where $u(t)\in\R$ is the input, $y(t)\in\R$ is the output, $M=0.5$, $m=0.2$, $b=0.1$, $l=0.2$, $I=0.006$, and $g=9.8$.
The plant \eqref{eq:plant} is obtained by discretizing \eqref{eq:invpen} under the sampling period $50$ ms, as
\begin{equation}\label{eq:simul plant}
\begin{aligned}
	\!\!\!D_p(z)&\!=\!z^4\!-\!4.0757z^3\!+\!6.1423z^2\!-\!4.0581z\!+\!0.9915,\\
	\!\!\!N_p(z)&\!=\!0.0021z^3\!-\!0.0023z^2\!-\!0.0023z\!+\!0.0021.
\end{aligned}
\end{equation}
At Step~2 of Algorithm~\ref{alg:stabil}, $\gamma^\ini(z)$ is chosen to have 
$-0.2616$, $0.3728$, $0.6769\pm0.6490i$, $0.9168\pm 0.1990i$, and $0.9650\pm 0.1i$ as its roots.
At Step~5, we set $\mu=0.99$ and $x^\star=\lceil x_0\rfloor$.
In this example, the control input \eqref{eq:u} achieves $x_k=x^\star$ when $k=1$.
Accordingly, we obtain the controller \eqref{eq:ctr}
as
\begin{align}
D_c(z)\!&=\!z^4\left(z^4-z^3-13z^2-4z+10\right),\label{eq:ex_ctr}\\
N_c(z)\!&=\!10^3\left(-6.4046z^7+17.154z^6-14.891z^5+3.8949z^4\right.\notag\\
&\,\,\left.+0.27228z^3\right)-6.0466z^2-23.6399z+4.7839,\notag
\end{align}
which yields a stable closed-loop system; the maximum absolute value of the roots of \eqref{eq:cl} is $0.9701$.
The code for this example is available as \verb|integer_ctr/stabilization.m| at \url{https://github.com/CDSL-EncryptedControl/CDSL}.

\begin{remark}
In general, Algorithm~\ref{alg:stabil} returns a controller of relatively higher order than typical stabilizing controllers, which may lead to increased computation load when implemented over encrypted data in a naive manner.
To alleviate this burden, existing ``packing'' methods can be employed, as in \cite{RGSW,RCF}.
For example,
the controller \eqref{eq:ex_ctr} was implemented in \cite{RCF} with a conservative encryption security level, and the resulting computation time per time step was below $6$ ms.
\end{remark}

\section{Application to Conversion Problem}\label{sec:conv}

This section addresses the conversion problem, where a pre-designed controller is given and the objective is to design an alternative controller
having integer coefficients that preserves the performance of the pre-designed controller in a certain sense.
As in the previous result \cite{CDC23},
we consider a reference signal injected to the controller as an input, as depicted in Fig.~\ref{fig:cl},
and aim to preserve the transfer function of the closed-loop system from the reference to the plant output exactly.

\begin{figure}
\centering
\begin{subfigure}[b]{0.47\linewidth}
	\centering
	\resizebox{\linewidth}{!}{\begin{tikzpicture}[node distance=1.5cm,>=latex']
			\node[rectangle,draw=blue, minimum height = 2.2cm, minimum width = 4.1cm, very thick, dashed, label = 260:{\large $ T_{ry}(z)$}] (cl) at (2.6,-0.3) {};
			\node[coordinate, name = reference]{};
			\node[block, right of = reference, minimum width = 1cm, name = controller, thick]{\large $C$};
			\node[coordinate, right of = controller, node distance = 0.5cm, name = input] {};
			\node[block, right of = input, minimum width = 1cm, name = plant, thick]{\large $P$};
			\node[coordinate, right of = plant, node distance = 1.8cm, name = output]{};
			\node[coordinate, right of = plant, node distance = 1cm, name = outputJunc]{};
			\node[coordinate, below of = outputJunc, node distance = 1.2cm, name = temp]{};
			\draw[->, thick] (reference) -- (controller) node[pos = 0.3, above]{\large $r$};
			\draw[-, thick] (controller) -- (input) node[]{};
			\draw[->, thick] (input) -- (plant) node[pos = 0.5, above]{\large $u$};
			\draw[->, thick] (plant) -- (output) node[pos = 0.7, above]{\large $y$};
			\draw[-, thick] (outputJunc) -- (temp);
			\draw[->, thick] (temp) -| (controller) node[]{};
	\end{tikzpicture}}
	\caption{Before conversion}
\end{subfigure}
\hfill
\begin{subfigure}[b]{0.47\linewidth}
	\centering
	\resizebox{\linewidth}{!}{\begin{tikzpicture}[node distance=1.5cm,>=latex']
			\node[rectangle,draw=red, minimum height = 2.2cm, minimum width = 4.1cm, very thick, dashed, label = 272:{\large $ T_{ry}^\prime(z)=T_{ry}(z)$}] (cl) at (2.6,-0.3) {};
			\node[coordinate, name = reference]{};
			\node[block, right of = reference, minimum width = 1cm, name = controller, thick]{\large $C^\prime$};
			\node[coordinate, right of = controller, node distance = 0.5cm, name = input]{};
			\node[block, right of = input, minimum width = 1cm, name = plant, thick]{\large $P$};
			\node[coordinate, right of = plant, node distance = 1.8cm, name = output]{};
			\node[coordinate, right of = plant, node distance = 1cm, name = outputJunc]{};
			\node[coordinate, below of = outputJunc, node distance = 1.2cm, name = temp]{};
			\draw[->, thick] (reference) -- (controller) node[pos = 0.3, above]{\large $r$};
			\draw[-, thick] (controller) -- (input) node[]{};
			\draw[->, thick] (input) -- (plant) node[pos = 0.5, above]{\large $u$};
			\draw[->, thick] (plant) -- (output) node[pos = 0.7, above]{\large $y$};
			\draw[-, thick] (outputJunc) -- (temp);
			\draw[->, thick] (temp) -| (controller) node[]{};
	\end{tikzpicture}}
	\caption{After conversion}
\end{subfigure}
\caption{Conversion of a pre-designed controller $C$ to a new controller $C^\prime$ having integer coefficients.}
\label{fig:cl}
\end{figure}

Throughout this section, $C(z)$ denotes the proper transfer function matrix of the pre-designed controller, written by
\begin{equation}\label{eq:pre ctr}
C(z)=\frac{1}{D_c(z)}\begin{bmatrix}
	N_{c,y}(z) & N_{c,r}(z)
\end{bmatrix},
\end{equation}
where $D_c(z)$ is a monic polynomial and the first and the second inputs are the plant output and the reference, respectively.
Let this pre-designed controller stabilize the given plant \eqref{eq:plant}.
Then, the closed-loop transfer function that we aim to preserve is
\begin{equation}\label{eq:cltf}
T_{ry}(z)=\frac{N_{c,r}(z)N_p(z)}{D_p(z)D_c(z)-N_p(z)N_{c,y}(z)}.
\end{equation}

The problem is to design a new controller
\begin{equation}\label{eq:new ctr}
C^\prime(z)=\frac{1}{D_c^\prime(z)}\begin{bmatrix}
	N_{c,y}^\prime(z) & N_{c,r}^\prime(z)
\end{bmatrix}
\end{equation}
such that i) $D_c^\prime(z)$ is an integer monic polynomial, ii) the closed-loop system of \eqref{eq:plant} and \eqref{eq:new ctr} is internally stable, and iii) the transfer function from the reference to the plant output, denoted by $T_{ry}^\prime(z)$, is equal to \eqref{eq:cltf}.

A method to design such a controller is to first solve the following subproblem \cite{CDC23}, which is similar to Problem~\ref{prob:stabilization}.
\begin{problem}\label{prob:conversion}
Find polynomials $\alpha(z)$, $\beta(z)$, and $\gamma(z)$ such that
\begin{equation}\label{eq:conversion}
	\alpha(z)D_c(z)+\beta(z)N_p(z)=\gamma(z)
\end{equation}
and satisfy the followings:
\begin{enumerate}
	\item[(C1)] $\alpha(z)$ is a Schur stable monic polynomial.
	\item[(C2)] $\gamma(z)$ is an integer monic polynomial.
	\item[(C3)] $\deg(\beta(z))<\deg(\gamma(z))-n$.
\end{enumerate}
\end{problem}

Then, a new controller \eqref{eq:new ctr} can be designed from a solution to Problem~\ref{prob:conversion} as in \cite{CDC23}, as
\begin{equation}\label{eq:new design}
\begin{aligned}
	D_c^\prime(z)&=\gamma(z),\quad N_{c,r}^\prime(z)=\alpha(z)N_{c,r}(z),\\
	N_{c,y}^\prime(z)&=\beta(z)D_p(z)+\alpha(z)N_{c,y}(z).
\end{aligned}
\end{equation}
This ensures that the new controller has integer coefficients by (C2), achieves $T_{ry}(z)=T_{ry}^\prime(z)$ by \eqref{eq:conversion}, keeps the internal stability by (C1), and is proper by (C3).
However, the previous result \cite{CDC23} solves Problem~\ref{prob:conversion} only when the numerator $N_p(z)$ of the plant \eqref{eq:plant} is a constant.
In contrast,
we propose a method to solve Problem~\ref{prob:conversion} in general, given that the polynomials $D_c(z)$ and $N_p(z)$ are coprime.

As seen from the resemblance of Problem~\ref{prob:conversion} to Problem~\ref{alg:stabil},
the proposed method is based on the principle of Section~\ref{sec:stabil}.
This time, $(\alpha(z),\beta(z),\gamma(z))$ is iteratively updated so that $\gamma(z)$ eventually becomes an integer polynomial.
Analogously to \eqref{eq:pm},
we use the fact that given $(\alpha(z),\beta(z),\gamma(z))$ satisfying \eqref{eq:conversion}, the next $(\alpha^+(z),\beta^+(z),\gamma^+(z))$ can be constructed as
\begin{multline*}
\underbrace{a(z)\alpha(z)}_{=\alpha^+(z)}D_c(z)+\underbrace{\left(a(z)\beta(z)+w(z)\right)}_{=\beta^+(z)}N_p(z)\\
=\underbrace{a(z)\gamma(z)+w(z)N_p(z)}_{=\gamma^+(z)}
\end{multline*}
with some polynomials $a(z)$ and $w(z)$, so that \eqref{eq:conversion} is met.

\begin{algorithm}[t]
\caption{Solving Problem~\ref{prob:conversion}.}\label{alg:conv}
\begin{algorithmic}[1]
	\renewcommand{\algorithmicrequire}{\textbf{Input:}}
	\Require $D_c(z),\,N_p(z),\,n$.
	\State Choose a Schur stable monic polynomial $\alpha(z)$ coprime to $N_p(z)$ such that
	$\deg(\alpha(z))\geq n-\deg(D_c(z))$.
	\State $N\gets \deg(\alpha(z)D_c(z))-n$.
	\State $r(z)\gets \calF(z^N,\alpha(z)D_c(z))$.
	\State
	Let $x_0\in\R^n$ such that $p_{x_0}(z)=r(z)$. $k\gets 0$.
	\State Perform Step~4 of Algorithm~\ref{alg:stabil}.
	\While{$x_k\neq x^\star$}
	\State Perform Steps~6 and 7 of Algorithm~\ref{alg:stabil}.
	\State $\alpha(z)\gets p_{u_k}(z)\alpha(z)$, $N\gets N+n$, $k\gets k+1$.
	\EndWhile
	\State $\gamma(z)\gets z^Np_{x^\star}(z)$.
	\State $\beta(z)\gets (\gamma(z)-\alpha(z)D_c(z))/N_p(z)$.
	\renewcommand{\algorithmicrequire}{\textbf{Output:}}
	\Require $\alpha(z),\,\beta(z),\,\gamma(z)$.
\end{algorithmic}
\end{algorithm}

We provide the complete method as Algorithm~\ref{alg:conv}.
From the results of Section~\ref{sec:stabil},
it can be verified that Algorithm~\ref{alg:conv} returns a solution to Problem~\ref{prob:conversion}, leading to the following theorem.

\begin{theorem}
Given a plant \eqref{eq:plant} and a pre-designed controller \eqref{eq:pre ctr}, suppose that $D_c(z)$ and $N_p(z)$ are coprime.
Then, there exists a controller \eqref{eq:new ctr} such that the followings hold:
\begin{enumerate}[leftmargin=*]
	\item $D_c^\prime(z)$ is an integer monic polynomial.
	\item
	$
	T_{ry}^\prime(z)=T_{ry}(z)$.
	\item The closed-loop system of \eqref{eq:plant} and \eqref{eq:new ctr} is internally stable.
\end{enumerate}
Furthermore, such a controller can be constructed from the outputs of Algorithm~\ref{alg:conv}, according to \eqref{eq:new design}.
\end{theorem}

\begin{proof}
It suffices to show that Algorithm~\ref{alg:conv} solves Problem~\ref{prob:conversion}.
After Step~4, $p_{x_0}(z)$ and $N_p(z)$ are coprime since $\alpha(z)$ and $N_p(z)$ are coprime.
Then, by Proposition~\ref{prop:finite}, Steps~7 and 8 are repeated only a finite number of times.
Thus, (C1) and (C2) hold by construction.
It is derived that (C3) and \eqref{eq:conversion} hold
by showing that $p_{x_k}(z)=\calF(z^N,\alpha(z)D_c(z))$ after each iteration.
The rest of the proof is analogous to that of Theorem~\ref{thm:stabil}.
\end{proof}

\subsection{Numerical example}

We demonstrate the proposed conversion method through a numerical example.
The plant \eqref{eq:plant} is the linearized inverted pendulum model \eqref{eq:simul plant}, and
the controller \eqref{eq:pre ctr} is designed as
\begin{equation}\label{eq:simul pre}
\begin{aligned}
	D_c(z)&=z^5-2.8826z^4+0.1067z^3+0.4848z^2\\
	&\,\,+3.8324z-2.5413,\\
	N_{c,y}(z)&=10^3\left(-1.556z^4+5.8219z^3-8.1324z^2\right.\\&\,\,+\left.5.0230z-1.1566\right),\\
	N_{c,r}(z)&=-0.02z^4+0.0566z^3-0.0584z^2\\
	&\,\,+0.0258z-0.0041.
\end{aligned}
\end{equation}
To convert \eqref{eq:simul pre}, we apply Algorithm~\ref{alg:conv}.
At Step~1, the initial $\alpha(z)$ is chosen
to have roots at $-0.7493$, $-0.1861$, $-0.2412\pm 0.8757i$, and $-0.1373\pm 0.9794i$.
By setting $\mu=0.99$ and $x^\star=\lceil x_0\rfloor$ at Step~5, $x_k$ reaches $x^\star$ within $4$ time steps.
As a result, we obtain the converted controller \eqref{eq:new design} from the outputs of Algorithm~\ref{alg:conv}, having integer coefficients as
\begin{equation*}
D_c^\prime(z)=z^{23}\left(z^4-z^3-4z^2-2z+4\right).
\end{equation*}

The codes used for conversion and the simulation shown in Fig.~\ref{fig:conversion} are available as \verb|integer_ctr/conversion.m| and \verb|integer_ctr/conversion_simul.m|, respectively, at \url{https://github.com/CDSL-EncryptedControl/CDSL}.

\begin{figure}
\centering
\begin{subfigure}[b]{0.49\linewidth}
	\centering
	\input{output}
	\caption{Cart position}
	\label{subfig:output}
\end{subfigure}
\hfill
\begin{subfigure}[b]{0.49\linewidth}
	\centering
	\input{angle}
	\caption{Angle of the pendulum}
	\label{subfig:angle}
\end{subfigure}
\caption{Performance of the pre-designed controller $C$ of \eqref{eq:simul pre} and the converted controller $C^\prime$.}
\label{fig:conversion}
\end{figure}

\begin{remark}
    Although the proposed conversion preserves the steady-state response of the closed-loop system perfectly, the transient response is generally altered;
    this is due to the pole-zero cancellations occurred by $\alpha(z)$ in the closed-loop transfer function $T_{ry}^\prime(z)=T_{ry}(z)\frac{\alpha(z)}{\alpha(z)}$.
    Nevertheless, one can always adjust the initial condition of the new controller to reproduce the original transient response,
    given that the plant initial state (or its estimate) is known.
    Fig.~\ref{fig:conversion} is a result of such adjustment.
\end{remark}

\section{Conclusion}\label{sec:conclude}

We have shown that it is always possible to design a controller consisting of integer coefficients that stabilizes a given linear plant.
An algorithm to design such controller is provided in a constructive way, which aids those who are interested in implementing encrypted control systems.
Moreover, we have proposed a method to convert a pre-designed controller to have integer coefficients,
which yields the same transfer function of the closed-loop system as before.

\bibliographystyle{IEEEtran}
\bibliography{ref}

\end{document}